\documentclass{scrartcl}
\usepackage{fullpage}
\usepackage[utf8]{inputenc} 
\usepackage[T1]{fontenc}    
\usepackage{amsfonts}       
\usepackage{amsmath}
\usepackage{mathtools}
\usepackage{amsthm}
\usepackage{bm}
\usepackage{color}
\usepackage{soul}           
\usepackage{url}
\usepackage{booktabs}       
\usepackage{nicefrac}       
\usepackage{enumerate}
\usepackage{graphicx}
\DeclareGraphicsExtensions{.pdf,.png,.jpg}
\graphicspath{{figures/}}
\usepackage{comment}
\usepackage{hyperref}

\newcommand{\R}{\mathbb{R}}

\newcommand{\BigO}[1]{\ensuremath{\mathcal{O}\left(#1\right)}} 
\newcommand{\BigOm}[1]{\ensuremath{\Omega\left(#1\right)}} 
\newcommand{\BigT}[1]{\ensuremath{\Theta\left(#1\right)}} 
\newcommand{\poly}{\mathrm{poly}}
\newcommand{\vect}[1]{\ensuremath{\mathbf{#1}}}            
\newcommand{\mat}[1]{\ensuremath{\mathbf{\MakeUppercase{#1}}}}   
\newcommand{\Exp}[2]{\ensuremath{\mathbb{E}_{#1}\left[#2\right]}}                
\newcommand{\Ind}[1]{\ensuremath{\mathbf{1}\left[#1\right]}}                     
\newcommand{\Norm}[1]{\ensuremath{\left\lVert #1 \right\rVert}}                  
\newcommand{\NormI}[1]{\ensuremath{\left\lVert #1 \right\rVert}_1}                 
\newcommand{\NormII}[1]{\ensuremath{\left\lVert #1 \right\rVert}_2}                
\newcommand{\NormInfty}[1]{\ensuremath{\left\lVert #1 \right\rVert_{\infty}}}    

\newcommand{\InNorm}[1]{{\left\vert\kern-0.2ex\left\vert\kern-0.2ex\left\vert #1 
    \right\vert\kern-0.2ex\right\vert\kern-0.2ex\right\vert}}                    

\newcommand{\InNormII}[1]{{\left\vert\kern-0.2ex\left\vert\kern-0.2ex\left\vert #1 
    \right\vert\kern-0.2ex\right\vert\kern-0.2ex\right\vert}_2}                    

\newcommand{\InNormInfty}[1]{{\left\vert\kern-0.2ex\left\vert\kern-0.2ex\left\vert #1 
    \right\vert\kern-0.2ex\right\vert\kern-0.2ex\right\vert}_{\infty}}           

\newcommand{\Abs}[1]{\ensuremath{\left \lvert #1 \right \rvert}}                 
\newcommand{\Prob}[1]{\ensuremath{\mathrm{Pr}\left\{ #1 \right\}}}               
\newcommand{\Grad}{\nabla}                                                       

\newcommand{\Land}{\wedge}                                                        
\newcommand{\defeq}{\doteq}                                         

\DeclareMathOperator*{\argmin}{argmin}
\DeclareMathOperator*{\argmax}{argmax}
\DeclareMathOperator*{\sign}{sign}

\newtheorem{assumption}{Assumption}
\newtheorem{lemma}{Lemma}
\newtheorem{theorem}{Theorem}
\newtheorem{remark}{Remark}

\newcommand{\x}{\mathbf{x}}
\newcommand{\w}{\mathbf{w}}
\newcommand{\win}{\ensuremath{\w_{-i}}} 
\newcommand{\wtin}{\ensuremath{{\w^*_{-i}}}} 
\newcommand{\whin}{\ensuremath{\widehat{\w}_{-i}}} 
\newcommand{\vv}{\ensuremath{\mathbf{v}}} 
\newcommand{\vi}{\ensuremath{\mathbf{v}_i}} 
\newcommand{\vhi}{\ensuremath{\mathbf{\widehat{v}}_i}} 
\newcommand{\vt}{\ensuremath{\mathbf{v^*}}} 
\newcommand{\vti}{\ensuremath{{\mathbf{v}^*_i}}} 
\newcommand{\vh}{\ensuremath{\widehat{\mathbf{v}}}} 
\newcommand{\zv}{\ensuremath{\vect{z}}} 
\newcommand{\zl}{\ensuremath{{\vect{z}^{(l)}}}} 
\newcommand{\zi}{\ensuremath{\vect{z}_i}} 
\newcommand{\zil}{\ensuremath{\vect{z}_i^{(l)}}} 
\newcommand{\zsl}{\ensuremath{{\vect{z}^{(l)}_S}}} 
\newcommand{\zs}{\ensuremath{\vect{z}_S}} 
\newcommand{\xin}{\ensuremath{\x_{-i}}} 
\newcommand{\NE}{\ensuremath{\mathcal{NE}}}
\newcommand{\Data}{\ensuremath{\mathcal{D}}}
\newcommand{\Game}{\ensuremath{\mathcal{G}}}
\newcommand{\Sc}{\ensuremath{S^c}} 
\newcommand{\vs}{\ensuremath{\vect{v}_{S}}} 
\newcommand{\vhs}{\ensuremath{\vh_{S}}} 
\newcommand{\vts}{\ensuremath{{\vect{v}^*_{S}}}} 
\newcommand{\loss}{\ensuremath{\ell}} 
\newcommand{\Hess}{\ensuremath{\nabla^2}} 
\newcommand{\Rad}{\mathfrak{R}} 
\newcommand{\ds}{\Delta_S} 
\newcommand{\dhs}{\widehat{\Delta}_S} 
\newcommand{\dbs}{\Delta_{\overline{S}}} 
\newcommand{\dbsc}{\Delta_{\overline{\Sc}}} 
\newcommand{\yv}{\vect{y}} 
\newcommand{\mumax}{\mu_{\mathrm{max}}}
\newcommand{\mumin}{\mu_{\mathrm{min}}}
\newcommand{\eigmin}{\lambda_{\mathrm{min}}}
\newcommand{\eigmax}{\lambda_{\mathrm{max}}}

\newcommand{\cmin}{C_{\mathrm{min}}}
\newcommand{\dmax}{D_{\mathrm{max}}}

\newcommand{\A}{\mathcal{A}} 
\newcommand{\U}{\mathcal{U}} 
\newcommand{\Nh}{\mathcal{N}} 
\newcommand{\rhomin}{\rho_{\mathrm{min}}}
\newcommand{\vtbs}{\vect{v}^*_{\overline{S}}}

\begin{document}

\title{From Behavior to Sparse Graphical Games: Efficient Recovery of Equilibria}

\author{Asish Ghoshal and Jean Honorio\\
Department of Computer Science\\
Purdue University\\
West Lafayette, IN - 47906\\
\{aghoshal, jhonorio\}@purdue.edu}

\date{}

\maketitle

\begin{abstract}
In this paper we study the problem of exact recovery of the pure-strategy
Nash equilibria (PSNE) set of a graphical game 
from noisy observations of joint actions of the players alone.
We consider sparse linear influence games ---
a parametric class of graphical games with linear payoffs, and
represented by directed graphs of $n$ nodes (players) and in-degree of at most $k$.
We present an $\ell_1$-regularized logistic regression
based algorithm for recovering the PSNE set exactly, that is both 
computationally efficient --- i.e. runs in polynomial time --- 
and statistically efficient --- i.e. has logarithmic sample complexity. Specifically,
we show that the sufficient number of samples required for exact PSNE recovery
scales as $\BigO{\poly(k) \log n}$. We also validate our theoretical results
using synthetic experiments. 
\end{abstract}
\section{Introduction and Related Work}
\label{sec:introduction}
Non-cooperative game theory is widely regarded as an appropriate
mathematical framework for studying \emph{strategic} behavior in multi-agent scenarios.
The core solution concept of \emph{Nash equilibrium} 
describes the stable outcome of the overall behavior of self-interested agents  
--- for instance people, companies, governments, groups or autonomous systems ---
interacting strategically with each other and in distributed settings.

Over the past few years, considerable progress has been made in analyzing 
behavioral data using game-theoretic tools, e.g. 
computing Nash equilibria \cite{blum2006continuation,ortiz2002nash,vickrey2002multi},
most influential agents \cite{irfan14}, price of anarchy \cite{ben2011local} 
and related concepts in the context of graphical games.
In \emph{political science} for instance, Irfan and Ortiz \cite{irfan14} 
identified, from congressional voting records, the most influential senators in the U.S. congress --- 
a small set of senators whose collective behavior forces every other senator to a unique choice of vote.
Irfan and Ortiz \cite{irfan14} also observed that 
the most influential senators were strikingly similar to the gang-of-six senators, 
formed during the national debt ceiling negotiations of 2011.
Further, using graphical games,  Honorio and Ortiz \cite{honorio15} 
showed that Obama's influence on Republicans increased in the last sessions before candidacy, 
while McCain's influence on Republicans decreased.

The problems in \emph{algorithmic game theory} described above, i.e. 
computing the Nash equilibria, computing the price of anarchy or finding the most influential agents,
require a known graphical game which is not available apriori in real-world settings.
Therefore, Honorio and Ortiz \cite{honorio15} proposed learning graphical games from behavioral data,
using maximum likelihood estimation (MLE) and \emph{sparsity}-promoting methods.
Honorio and Ortiz \cite{honorio15} and Irfan and Ortiz 
\cite{irfan14} have also demonstrated the usefulness of learning \emph{sparse} graphical games from behavioral data
in real-world settings, through their analysis of
the voting records of the U.S. congress as well as the U.S. supreme court.

In this paper, we analyze a particular method proposed by Honorio and Ortiz \cite{honorio15} for learning
sparse linear influence games, namely, using logistic regression for learning
the neighborhood of each player in the graphical game, independently. Honorio and Ortiz \cite{honorio15}
showed that the method of independent logistic regression is likelihood consistent; i.e. in the
infinite sample limit, the likelihood estimate converges to the true achievable likelihood. In
this paper we obtain the stronger guarantee  of recovering the true PSNE set exactly. Needless
to say, our stronger guarantee comes with additional conditions that the true game must satisfy
in order to allow for exact PSNE recovery. The most crucial among our conditions is the assumption
that the minimum payoff across all players and all joint actions in the true PSNE set is strictly positive.
We show, through simulation experiments, that our assumptions indeed bears out for large classes of graphical games, where we
are able to exactly recover the PSNE set. 

Finally, we would like to draw the attention of the reader to the fact that $\ell_1$-regularized logistic
regression has been analyzed by Ravikumar et. al. \cite{Ravikumar2010} in the context of learning
sparse Ising models. Apart from technical differences and differences in proof techniques,
our analysis of $\ell_1$-penalized logistic regression
for learning sparse graphical games differs from Ravikumar et. al. \cite{Ravikumar2010} conceptually --- in the sense
that we are not interested in recovering the edges of the true game graph, but only the PSNE set. Therefore,
we are able to avoid some stronger and non-intuitive conditions required by Ravikumar et. al. \cite{Ravikumar2010}, 
such as mutual incoherence.

The rest of the paper is organized as follows.
In section \ref{sec:prelim} we provide a brief overview of graphical games and, specifically, linear influence games.
We then formalize the problem of learning linear influence games in section \ref{sec:problem_formulation}.
In section \ref{sec:results}, we present our main method and associated theoretical results. Then, in
section \ref{sec:experiments} we provide experimental validation of our theoretical guarantees. Finally,
in section \ref{sec:conclusion} we conclude by discussing avenues for future work.

\section{Preliminaries}
\label{sec:prelim}
In this section we review key concepts behind graphical games introduced by Kearns et. al. \cite{kearns01}
and linear influence games (LIGs) introduced by Irfan and Ortiz \cite{irfan14} and
Honorio and Ortiz \cite{honorio15}.
\subsection{Graphical Games}
A \emph{normal-form game} $\Game$ in classical game theory is defined by the triple
$\Game = (V, \A, \U)$ of players, actions and payoffs. $V$ is the set of players, 
and is given by the set $V = \{1,\ldots,n\}$, if there are $n$ players. $\A$
is the set of actions or \emph{pure-strategies} and is given by the Cartesian product
$\A \defeq \times_{i \in V} A_i$, where $A_i$ is the set of pure-strategies of the $i$-th player.
Finally, $\U \defeq \{u_i\}_{i=1}^n$, is the set of payoffs,
 where $u_i: A_i \times_{j \in V\setminus i} A_j \rightarrow \R$ specifies the payoff for the $i$-th player given
 its action and the joint actions of the all the remaining players. 

\sloppy
A key solution concept in non-cooperative game theory is the \emph{Nash equilibrium}. 
For a non-cooperative game, a joint action $\x^* \in \A$ is a pure-strategy Nash equilibrium (PSNE)
if, for each player $i$, $x_i^* \in \argmax_{x_i \in A_i} u_i(x_i, \x^*_{-i})$, where 
$\x^*_{-i} = \{x^*_j \vert j \neq i\}$. In other words,
$\x^*$ constitutes the mutual best-response for all players and no player has any incentive to 
unilaterally deviate from their optimal action $x^*_i$ given the joint actions of the remaining players $\x^*_{-i}$. 
The set of all \emph{pure-strategy Nash equilibrium} (PSNE) for a game $\Game$ is defined as follows:
\begin{align}
\NE(\Game) = \left\{\x^* \big| (\forall i \in V)\; x^*_i \in 
	\argmax_{x_i \in A_i} u_i(x_i, \x^*_{-i}) \right\}. \label{eq_psne_set}
\end{align}
\fussy

\emph{Graphical games}, introduced by Kearns et. al. \cite{kearns01},
extend the formalism of \emph{Graphical models} to games. That is, a graphical game 
$G$ is defined by the \emph{directed graph}, $G = (V, E)$, of vertices and 
directed edges (arcs), where vertices correspond to players and arcs
encode ``influence'' among players i.e. the payoff of  
the $i$-th player only depends on the actions of its (incoming) neighbors.

\subsection{Linear Influence Games}
Linear influence games (LIGs), introduced by Irfan and Ortiz \cite{irfan14} and Honorio and Ortiz \cite{honorio15},
are graphical games with binary actions, or pure strategies, and parametric (linear) payoff functions.
We assume, without loss of generality, that the joint action space $\A = \{-1, +1\}^n$.
A linear influence game between $n$ players, $\Game(n) = (\mat{W}, \vect{b})$, is characterized by
(i) a matrix of weights $\mat{W} \in \R^{n \times n}$, where the entry $W_{ij}$ indicates the amount of influence (signed) that the $j$-th player has on the $i$-th player and (ii) a bias vector $\vect{b} \in \R^n$, where $b_i$ captures
the prior preference of the $i$-th player for a particular action $x_i \in \{-1, +1\}$. The
payoff of the $i$-th player given the actions of the remaining players is then given as
$u_i(x_i, \xin) = x_i(\win^T\xin - b_i) $, and the PSNE set is defined as follows:
\begin{align}
\NE(\Game(n)) = \left\{\x | (\forall i)\; x_i(\win^T\xin - b_i) \geq 0 \right\} \label{eq:lig_psne},
\end{align}
where $\win$ denotes the $i$-th row of $\mat{W}$ without the $i$-th entry, i.e.
$\win = \{w_{ij} \vert j \neq i \}$. Note that we have $\mathrm{diag}(\mat{W}) = 0$.
For linear influence games $\Game(n)$, we can
define the neighborhood and signed neighborhood of the $i$-th vertex as 
$\Nh(i) = \{j | \Abs{w_{i,j}} > 0\}$ and $\Nh_{\pm}(i) = \{\sign(w_{ij}) | j \in \Nh(i)\}$
respectively.
It is important to note that we don't include the outgoing arcs in our definition of
the neighborhood of a vertex.
Thus, for linear influence games, the weight matrix $\mat{W}$ and 
the bias vector $\vect{b}$, completely specify the game and the PSNE set
induced by the game. Finally, let $\Game(n, k)$ denote sparse games 
over $n$ players where the in-degree of any vertex is at most $k$, i.e. for all $i$, $\Abs{\Nh_i} \leq k$.

\section{Problem Formulation}
\label{sec:problem_formulation}
Having introduced the necessary definitions, we now introduce the problem of learning
sparse linear influence games from observations of joint actions only. We assume that
there exists a game $\Game^*(n, k) = (\mat{W}^*, \vect{b}^*)$ from which a ``noisy'' 
data set $\Data = \{\x^{(l)}\}_{l=1}^m$ of $m$ observations is generated,
where each observation $\x^{(l)}$ is sampled independently and
identically from the following distribution:
\begin{align}
p(\x) = \frac{q\Ind{\x \in \NE^*}}{|\NE^*|} 
	+ \frac{(1 - q)\Ind{\x \notin \NE^*}}{2^n - |\NE^*|} \label{eq:obs_model}.
\end{align}
In the above distribution, $q$ is the probability of observing a data point from the set of Nash equilibria
and can be thought of as the ``signal'' level in the data set, while $1 - q$ can be thought of as
the ``noise'' level in the data set. We use $\NE^*$ instead
of $\NE(\Game^*(n, k))$ to simplify notation. We further assume that the game is 
non-trivial\footnote{This comes from Definition 4 in \cite{honorio15}.}
i.e. $\Abs{\NE^*} \in \{1, \ldots 2^n - 1\}$ and that $q \in (\nicefrac{\Abs{\NE^*}}{2^n}, 1)$.
The latter assumption ensures that the signal level in the data set is more than the noise level
\footnote{See Proposition 5 and Definition 7 in \cite{honorio15} for a justification of this.}.
We define the equality of two games $\Game^*(n, k) = (\mat{W}^*, \vect{b}^*)$
and $\widehat{\Game}(n, k) = (\widehat{\mat{W}}, \widehat{\vect{b}})$ as follows:
\begin{gather*}
\Game^*(n, k) = \widehat{\Game}(n, k) \text{ iff } 
(\forall i)\; \Nh^*_{\pm}(i) = \widehat{\Nh}_{\pm}(i)\; \Land 
\sign(b^*_i) = \sign(\widehat{b}_i)\; \Land \; \NE^* = \widehat{\NE}.
\end{gather*}
A natural question to ask then is that, given only the data set $\Data$ and no other
information, is it possible to recover the weight matrix $\widehat{\mat{W}}$ 
and the bias vector $\widehat{\vect{b}}$ such that $\widehat{\Game}(n, k) = \Game^*(n, k)$?
Honorio and Ortiz \cite{honorio15} showed that it is in general
impossible to learn the true game $\Game^*(n, k)$ from observations of joint actions only
because multiple weight matrices $\mat{W}$ and bias vectors $\vect{b}$ can induce the same 
PSNE set and therefore have the same likelihood under the observation model \eqref{eq:obs_model} ---
an issue known as non-identifiablity in the statistics literature.
It is, however, possible to learn the equivalence class of games
that induce the same PSNE set. We define the equivalence of two games $\Game^*(n, k)$
and $\widehat{\Game}(n, k)$ simply as :
\begin{gather*}
\Game^*(n, k) \equiv \widehat{\Game}(n, k) \text{ iff } \NE^* = \widehat{\NE}.
\end{gather*}
Therefore, our goal in this paper is efficient 
and consistent recovery of the pure-strategy Nash equilibria set (PSNE) from observations
of joint actions only; i.e. given a data set $\Data$, drawn from some game $\Game^*(n, k)$
according to \eqref{eq:obs_model}, we infer a game $\widehat{\Game}(n, k)$ from $\Data$
such that $\widehat{\Game}(n, k) \equiv \Game^*(n, k)$.
\section{Method and Results}
\label{sec:results}
Our main method for learning the structure of a sparse LIG, $\Game^*(n,k)$,
is based on using $\ell_1$-regularized logistic regression, to learn
the parameters $(\win, b_i)$ for each player $i$ independently.
We denote by $\vi(\mat{W}, \vect{b}) = (\win, -b_i)$ the parameter vector 
for the $i$-th player, which characterizes its
payoff; and by $\zi(\x) = (x_i \xin,~ x_i)$ the ``feature'' vector. In the rest of the paper we
use $\vi$ and $\zi$ instead of $\vi(\mat{W}, \vect{b})$ and $\zi(\x)$ respectively, to simplify notation.
Then, we learn the parameters for the $i$-th player as follows:
\begin{align}
	\vhi &= \argmin_{\vi} \loss(\vi, \Data) + \lambda \NormI{\vi} \label{eq:optimization} \\
	\loss(\vi, \Data) &= \frac{1}{m} \sum_{l=1}^m \log(1 + \exp(-\vi^T\zil)) \label{eq:loss}.
\end{align}
We then set $\whin = [\vhi]_{1:(n - 1)}$ and $\widehat{b}_i = -[\vhi]_n$, where the notation
$[.]_{i:j}$ denotes indices $i$ to $j$ of the vector.  
We show that, under suitable assumptions on the true game $\Game^*(n,k) = (\mat{W}^*, \vect{b}^*)$, the parameters $\widehat{\mat{W}}$ and $\widehat{\vect{b}}$ obtained using \eqref{eq:loss} induce the same PSNE set as the true game, i.e. $\NE(\mat{W}^*, \vect{b}^*) = 
\NE(\widehat{\mat{W}}, \widehat{\vect{b}})$.
Before presenting our main results, however, we state and discuss the assumptions
on the true game under which it is possible to recover the true PSNE set of the game
using our proposed method.
\subsection{Assumptions}
The success of our method hinges on certain assumptions on the structure of the underlying game. These assumptions are in addition to the assumptions imposed on the parameter $q$ and the number of Nash equilibria of the game $\Abs{\NE^*}$, as required by definition of a LIG. Since the assumptions are related to the Hessian of the loss function \eqref{eq:loss}, we take a moment to introduce the expressions of the gradient and the Hessian here. The gradient and Hessian of the loss
function for any vector $\vv$ and the data set $\Data$ is given as follows:
\begin{align}
\Grad \loss(\vect{v}, \Data) = \frac{1}{m} 
	\sum_{l=1}^m\left\{ \frac{- \zl}{1 + \exp(\vect{v}^T \zl)} \right\} \label{eq:grad} \\
\Hess \loss(\vect{v}, \Data) = \frac{1}{m}  	
	\sum _{l=1}^m \eta(\vv^T \zl) \zl \zl^T,
\end{align}
where $\eta(x) = \nicefrac{1}{(e^{x/2} + e^{-x/2})^2}$. Finally, $\mat{H}^m_i$ denotes
the sample Hessian matrix with respect to the $i$-th player and the true parameter $\vti$, and  $\mat{H}_i^*$ denotes it's expected value, i.e. $\mat{H}_i^* \defeq \Exp{\Data}{\mat{H}^m_i} = 
\Exp{\Data}{\Hess \loss(\vti, \Data)}$.
In subsequent sections we drop the notational dependence of $\mat{H}_i^*$ and $\zi$ on $i$ to simplify notation.
However, it should be noted that the assumptions are with respect to
each player and must hold individually for each player $i$.
Now we are ready to state and describe the implications of our assumptions.

The following assumption ensures that the expected
loss is strongly convex and smooth.
\begin{assumption}
\label{ass:eigenvalue}
Let $S$ be the support of the vector $\vv$, i.e. $S \defeq \{i \vert \Abs{v_i} > 0 \}$.
Then, there exists constants $\cmin > 0$ and $\dmax \leq |S|$ such that:
\begin{align*}
\eigmin(\mat{H}_{SS}^*) \geq \cmin \text{ and } \eigmax(\Exp{\x}{\zv_S \zv_S^T}) \leq \dmax,
\end{align*}
for all $i$, where $\eigmin(.)$ and $\eigmax(.)$ denote the minimum and maximum eigenvalues respectively and $\mat{H}^*_{SS} = \{H^*_{i,j} | i, j \in S \}$.
\end{assumption}
To better understand the implications of the above assumptions first
note that the distribution over joint actions in \eqref{eq:obs_model}, after some algebraic manipulations, can be written as follows:
\begin{align}
p(\x) = \left(\frac{q - \nicefrac{\Abs{\NE^*}}{2^n}}{1 - \nicefrac{\Abs{\NE^*}}{2^n}}\right) \frac{\Ind{\x \in \NE^*}}{|\NE^*|} 
 \left(\frac{1 - q}{1 - \nicefrac{|\NE^*|}{2^n}}\right) \frac{1}{2^n} \label{eq:trans_obs_model}.
\end{align}
Thus, we have that the distribution over joint actions is a mixture of two uniform distributions,
one over the number of Nash equilibria and the other being the Rademacher distribution over $n$ variables.
We denote the latter distribution by $\Rad^n$, i.e. $\x \sim \Rad^n \implies \x \in \{-1, +1\}^n \Land p(\x) = \nicefrac{1}{2^n}$
for all $\x$.
Therefore the expected Hessian matrix $\mat{H}^*$ decomposes as a convex combination of two other Hessian matrices:
\begin{align}
\mat{H}^* = \nu \mat{H}^{\NE^*} + (1 - \nu) \mat{H}^{\Rad}, \label{eq:hess_comb}
\end{align}
where we have defined:
\begin{align*}
\mat{H}^{\NE*} &\defeq \frac{1}{\Abs{\NE^*}} \sum_{\x \in \NE^*} \eta(\vt^T \zv) \zv \zv^T, \\
\mat{H}^{\Rad} &\defeq \Exp{\zv \sim \Rad^n}{\eta(\vt^T \zv) \zv \zv^T}, \\
\nu &\defeq \left(\frac{q - \nicefrac{\Abs{\NE^*}}{2^n}}{1 - \nicefrac{\Abs{\NE^*}}{2^n}}\right).
\end{align*}
Now, by concavity of $\eigmin(.)$ and the Jensen's inequality we have that
\begin{align*}
\cmin  &\geq \nu \eigmin(\mat{H}^{\NE^*}_{SS}) + (1 - \nu) \eigmin(\mat{H}^{\Rad}_{SS}) \\
&\geq \nu \eigmin(\mat{H}^{\NE^*}_{SS}) + (1 - \nu) \eta(\NormI{\vt}) \\
&\geq (1 - \nu) \eta(\NormI{\vt}) > 0,
\end{align*}
where the last line follows from the fact that $\mat{H}^{\NE^*}_{SS}$ is positive semi-definite.
In fact the minimum eigenvalue, $\eigmin(\mat{H}^{\NE^*}_{SS})$, is zero if $\Abs{\NE^*} < n$.
Thus, we have that our assumption of $q \in (\nicefrac{\NE^*}{2^n}, 1)$ automatically
implies that $\eigmin(\mat{H}_i^*) \geq \cmin > 0$. 
We can also verify that the maximum eigenvalue is bounded as follows:
\begin{align*}
\dmax &\leq \nu \eigmax\left(\frac{1}{\Abs{\NE^*}} \sum_{\x \in \NE^*} \zv_S \zv_S^T \right) + 
	 (1 - \nu) \eigmax(\Exp{\zv \sim \Rad^n}{\zv_S \zv_S^T}) \\
&\leq \nu \Abs{S} + (1 - \nu).
\end{align*}

The following assumption characterizes the minimum payoff in the Nash equilibria set.
\begin{assumption}
\label{ass:payoff}
The minimum payoff in the PSNE set, $\rhomin$, is strictly positive, specifically:
\begin{align*}
x_i(\wtin^T \xin - b_i) \geq \rhomin > \nicefrac{5 \cmin}{\dmax}
                && (\forall\; \vect{x} \in \NE^*).
\end{align*}
\end{assumption}
Note that as long as the minimum payoff is strictly positive, we
can scale the parameters $(\mat{W}^*, \vect{b}^*)$ by the constant
$\nicefrac{5 \cmin}{\dmax}$ to satisfy the condition: $\rhomin > \nicefrac{5 \cmin}{\dmax}$,
without changing the PSNE set or the likelihood of the data.
Indeed the assumption that the minimum payoff is strictly positive is
is unavoidable for exact recovery of the PSNE set in a 
noisy setting such as ours, because otherwise this is akin to exactly recovering
the parameters $\vv$ for each player $i$. For example, if $\x \in \NE^*$ is such that
$\vt^T \x = 0$, then it can be shown that even if $\NormInfty{\vt - \vh} = \varepsilon$,
for any $\varepsilon$ arbitrarily close to $0$, then $\vh^T \x < 0$ and therefore 
$\NE(\mat{W}^*, \vect{b}^*) \neq \NE(\widehat{\mat{W}}, \widehat{\vect{b}})$.
Next, we present our main theoretical results for learning LIGs.

\subsection{Theoretical Guarantees}
Our main strategy for obtaining exact PSNE recovery guarantees is to first show,
using results from random matrix theory, that given the 
assumptions on the eigenvalues of the population Hessian matrices, the assumptions hold in the finite sample case with high probability. Then, we exploit the convexity properties of the logistic loss function to show that the weight vectors learned using penalized logistic regression is ``close'' to the true weight vectors. By our assumption that the minimum payoff in the PSNE set is strictly greater than zero, we show that the weight vectors inferred from a finite sample of joint actions
induce the same PSNE set as the true weight vectors.

\subsubsection{Minimum and Maximum Eigenvalues of Finite Sample Hessian and Scatter Matrices}
The following technical lemma shows that the assumptions on the eigenvalues of the Hessian
matrices, hold with high probability in the finite sample case.
\begin{lemma}
\label{lemma:eigvalue}
If $\eigmin(\mat{H}^*_{SS}) \geq \cmin$ and $\eigmax(\Exp{\x}{\zv_S \zv_S^T}) \leq \dmax$
then we have that 
\begin{gather*}
\eigmin(\mat{H}^m_{SS}) \geq \frac{\cmin}{2} \text{ and }
\eigmax\left(\sum_{l=1}^m \zsl \zsl^T \right) \leq 2 \dmax
\end{gather*}
with probability at least 
\begin{align*}
1 - \Abs{S} \exp \left(\frac{-m \cmin}{2 \Abs{S}}\right) \text{ and }
1 - \Abs{S} \exp\left( \frac{- m (1 - \nu)}{4\Abs{S}} \right)
\end{align*}
respectively.
\end{lemma}
\begin{proof}
Let $\mumin \defeq \eigmin(\mat{H}^*_{SS})$ and $\mumax \defeq \eigmax(\Exp{\x}{\zv_S \zv_S^T})$.
First note that for all $\zv \in \{-1, +1\}^n$:
\begin{gather*}
\eigmax(\eta(\vts^T \zs) \zs \zs^T) \leq \frac{\Abs{S}}{4} \defeq R \\
\eigmax(\zs \zs^T) \leq \Abs{S} \defeq R'.
\end{gather*}
Using the Matrix Chernoff bounds from Tropp \cite{tropp2012user}(Theorem 1.1), we have that
\begin{align*}
\Prob{\eigmin(\mat{H}^m_{SS}) \leq (1 - \delta) \mumin} \leq
	 \Abs{S} \left[ \frac{e^{-\delta}}{(1 - \delta)^{1 - \delta}} \right]^{\frac{m\mumin}{R}}.
\end{align*}
Setting $\delta = \nicefrac{1}{2}$ we get that
\begin{align*}
\Prob{\eigmin(\mat{H}^m_{SS}) \leq \nicefrac{\mumin}{2}} 
	\leq  \Abs{S} \left[\sqrt{\frac{2}{e}}\right]^{\frac{4 m \mumin}{\Abs{S}}} 
	\leq \Abs{S} \exp \left(\frac{-m\cmin}{2 \Abs{S}}\right).
\end{align*}
Therefore, we have
\begin{align*}
\Prob{\eigmin(\mat{H}^m_{SS}) > \nicefrac{\cmin}{2}} > 
	1 - \Abs{S} \exp \left(\frac{-m\cmin}{2 \Abs{S}}\right).
\end{align*}
Next, we have that
\begin{align*}
\mumax &= \eigmax(\Exp{\x}{\zv_S \zv_S^T}) 
\geq \eigmin(\Exp{\x}{\zv_S \zv_S^T}) \\
&\geq \nu \eigmin\left(\frac{1}{\Abs{\NE^*}} \sum_{x \in \NE^*} \zs \zs^T\right) +
	(1 - \nu) \\
&\geq (1 - \nu). 
\end{align*}
Once again invoking Theorem 1.1 from \cite{tropp2012user} and setting $\delta = 1$ we have that
\begin{gather*}
\Prob{\eigmax \geq (1 + \delta)\mumax} 
\leq \Abs{S} \left[ \frac{e^\delta}{(1 + \delta)^{1 + \delta}}\right]^{\nicefrac{(m \mumax)}{R'}} \\
\implies \Prob{\eigmax \geq 2 \mumax} 
\leq \Abs{S} \left[ \frac{e}{4}\right]^{\nicefrac{(m \mumax)}{\Abs{S}}} 
\leq \Abs{S} \exp\left( \frac{- m \mumax}{4\Abs{S}} \right) 
\leq \Abs{S} \exp\left( \frac{- m (1 - \nu)}{4\Abs{S}} \right).
\end{gather*}
Therefore, we have that
\begin{align*}
\Prob{\eigmax < 2 \dmax} > 1 - \Abs{S} \exp\left( \frac{- m (1 - \nu)}{4\Abs{S}} \right).
\end{align*}
\end{proof}

\subsubsection{Recovering the Pure Strategy Nash Equilibria (PSNE) Set}
Before presenting our main result on the exact recovery of the PSNE set from noisy
observations of joint actions, we first present a few technical lemmas that would be
helpful in proving the main result.
The following lemma bounds the gradient of the loss function \eqref{eq:loss}
at the true vector $\vt$, for all players.
\begin{lemma}
\label{lemma:grad_bound}
With probability at least $1 - \delta$ for $\delta \in [0, 1]$, we have that
\begin{align*}
\NormInfty{\Grad \loss(\vt, \Data)} < \nu \kappa + \sqrt{\frac{2}{m} \log \frac{2n}{\delta}},
\end{align*}
where $\kappa = \nicefrac{1}{(1 + \exp(\rhomin))}$ and $\rhomin \geq 0$ is the minimum 
payoff in the PSNE set.
\end{lemma}
\begin{proof}
Let $\vect{u}^m \defeq \Grad \loss(\vt, \Data)$ and  $u^m_j$ denote the $j$-th index of
$\vect{u}^m$. We have that
\begin{align*} 
&(\forall\; j \in \Sc \Land j \neq n)\; \Exp{}{u^m_j} \\
&\quad= \frac{\nu}{\Abs{\NE^*}} \sum_{\vect{\x \in \NE^*}} \frac{z_j}{1 + \exp((\vts)^T \vect{z}_S)} + 
   (1 - \nu) \Exp{\zv \sim \Rad^n}{\frac{z_j}{1 + \exp((\vts)^T \vect{z}_S)}}  \\
&\quad= \frac{\nu}{\Abs{\NE^*}} \sum_{\vect{\x \in \NE^*}} \frac{z_j}{1 + \exp((\vts)^T \vect{z}_S)} + 
   (1 - \nu)\left\{ \Exp{\zv_S \sim \Rad^{\Abs{S}}}{\frac{1}{1 + \exp((\vts)^T \vect{z}_S)}} \Exp{z_j \sim \Rad}{z_j} \right\} \\
&\quad\leq \frac{\nu \kappa}{\Abs{\NE^*}} \sum_{\vect{\x \in \NE^*}} x_i x_j \leq \nu \kappa
\end{align*}
Similarly, 
\begin{align*}
&(\forall\; j \in S \Land j \neq n)\; \Exp{}{u^m_j} \\
&\quad= \frac{\nu}{\Abs{\NE^*}} \sum_{\vect{\x \in \NE^*}} \frac{z_j}{1 + \exp((\vts)^T \vect{z}_S)} + 
        (1 - \nu) \Exp{\zv \sim \Rad^n}{\frac{z_j}{1 + \exp((\vts)^T \vect{z}_S)}} \\
&\quad\leq \frac{\nu}{\Abs{\NE^*}} \sum_{\vect{\x \in \NE^*}} \frac{z_j}{1 + \exp((\vts)^T \vect{z}_S)} + 
            (1 - \nu) \Exp{z_j \sim \Rad}{z_j} \\
&\quad\leq \nu \kappa.
\end{align*}
Following the same procedure as above, it can be easily shown that the above bounds hold for the case $j = n$
as well.
Also note that $\Abs{u^m_j} \leq 1$. Therefore, by using the Hoeffding's inequality \cite{hoeffding1963probability}
and a union bound argument we have that:
\begin{align*}
&\Prob{\max_{j=1}^n \Abs{u^m_j - \Exp{}{u^m_j}} < t}  > 1 - 2ne^{\nicefrac{-mt^2}{2}} \\
&\implies \Prob{\NormInfty{\vect{u^m} - \Exp{}{\vect{u^m}}} < t}  > 1 - 2ne^{\nicefrac{-mt^2}{2}} \\
&\implies \Prob{\NormInfty{\vect{u^m}} - \NormInfty{\Exp{}{\vect{u^m}}} < t}  > 1 - 2ne^{\nicefrac{-mt^2}{2}} \\
&\implies \Prob{\NormInfty{\vect{u^m}} < \nu \kappa + t} > 1 - 2ne^{\nicefrac{-mt^2}{2}}. 
\end{align*}
Setting $2n\exp(\nicefrac{-mt^2}{2}) = \delta$, we prove our claim.
\end{proof}
A consequence of Lemma \ref{lemma:grad_bound} is that, even with an infinite 
number of samples, the gradient of the loss function at the true vector $\vt$ doesn't vanish.
Therefore, we cannot hope to recover the parameters of the true game perfectly even with an infinite 
number of samples.
In the following technical lemma we show that the optimal vector $\vh$
for the logistic regression problem is close to the true vector $\vt$ in the support set $S$ of $\vt$. 
Next, in Lemma \ref{lemma_l1norm_bound}, we bound the difference between the
true vector $\vt$ and the optimal vector $\vh$ in the non-support set. The lemmas
together show that the optimal vector is close to the true vector. 
\begin{lemma}
\label{lemma_l2norm_bound}
If the regularization parameter $\lambda$ 
satisfies the following condition:
\begin{gather*}
\lambda \leq \frac{5\cmin^2}{16 \Abs{S} \dmax} - \nu \kappa - \sqrt{\frac{2}{m} \log \frac{2n}{\delta}},
\end{gather*}
then 
\begin{align*}
\NormII{\vts - \vhs} \leq \frac{5 \cmin}{4 \sqrt{\Abs{S}} \dmax},
\end{align*}
with probability at least $1 - (\delta + \Abs{S} \exp (\nicefrac{(-m \cmin)}{2 \Abs{S}}) + 
\Abs{S} \exp (\nicefrac{(-m (1 - \nu))}{4 \Abs{S}}))$.
\end{lemma}
\begin{proof}
The proof of this lemma follows the general proof structure of Lemma 3 in \cite{Ravikumar2010}.
First, we reparameterize the $\ell_1$-regularized loss function
\begin{align*}
f(\vect{v}_S) = \loss(\vs) + \lambda \NormI{\vs}
\end{align*}
as the loss function $\widetilde{f}$, which gives the loss at a point
that is $\ds$ distance away from the true parameter $\vts$ as follows:
\begin{align*}
\widetilde{f}(\ds) = \loss(\vts + \ds) - \loss(\vts)
    + \lambda (\NormI{\vts + \ds} - \NormI{\vts}),
\end{align*}
where $\ds = \vect{v}_S - \vts$.
Also note that the loss function $\widetilde{f}$ is shifted such that
the loss at the true parameter $\vts$ is $0$, i.e. $\widetilde{f}(\vect{0}) = 0$.
Further, note that the function
$\widetilde{f}$ is convex and is minimized at $\dhs = \vhs - \vts$,
since $\vhs$ minimizes $f$. Therefore, clearly $\widetilde{f}(\dhs) \leq 0$.
Thus, if we can show that the function $\widetilde{f}$ is strictly positive
on the surface of a ball of radius $b$, then the point $\dhs$ lies
inside the ball i.e. $\Norm{\vhs - \vts}_2 \leq b$. Using the Taylor's theorem
we expand the first term of $\widetilde{f}$ to get the following:
\begin{align}
\widetilde{f}(\ds) &= \Grad \loss(\vts)^T \ds + \ds^T \Hess \loss(\vts + \theta \ds) \ds
   + \lambda (\NormI{\vts + \ds} - \NormI{\vts}), \label{eq:f}
\end{align}
for some $\theta \in [0, 1]$. Next, we lower bound each of the terms in \eqref{eq:f}.
Using the Cauchy-Schwartz inequality, the first term in \eqref{eq:f} is bounded as follows:
\begin{align}
\Grad\loss(\vts)^T\ds &\geq - \NormInfty{\Grad\loss(\vts)} \NormI{\ds} 
	\geq - \NormInfty{\Grad\loss(\vts)} \sqrt{|S|} \NormII{\ds} \notag \\
&\geq - b \sqrt{|S|} \left(\nu \kappa + \sqrt{\frac{2}{m} \log \frac{2n}{\delta}}\right), \label{eq:lb1}
\end{align}
with probability at least $1 - \delta$ for $\delta \in [0, 1]$.
It is also easy to upper bound the last term in equation \ref{eq:f},
using the reverse triangle inequality as follows:
\begin{align*}
\lambda \Abs{\NormI{\vts + \ds} - \NormI{\vts}} \leq \lambda \NormI{\ds}.
\end{align*}
Which then implies the following lower bound:
\begin{align}
\lambda (\NormI{\vts + \ds} - \NormI{\vts}) &\geq - \lambda \NormI{\ds} 
	\geq -\lambda \sqrt{|S|} \NormII{\ds} \notag \\
&= -\lambda \sqrt{|S|} b. \label{eq:lb2}
\end{align}
Now we turn our attention to computing a lower bound of the second term of \eqref{eq:f}, which
is a bit more involved.
\begin{align*}
\ds^T \Hess \loss(\vts + \theta \ds) \ds
&\geq \min_{\NormII{\ds} = b} \ds^T \Hess \loss(\vts + \theta \ds) \ds \\
&= b^2 \eigmin(\Hess \loss(\vts + \theta \ds)).
\end{align*}
Now,
\begin{align*}
\eigmin(\Hess \loss(\vts + \theta \ds)) 
&\geq \min_{\theta \in [0, 1]} \eigmin\left(\Hess \loss(\vts + \theta \ds)\right) \\
&=\min_{\theta \in [0, 1]} \eigmin\left(\frac{1}{m} \sum_{l=1}^m \eta((\vts + \theta \ds)^T \zsl) \zsl (\zsl)^T\right).
\end{align*}
Again, using the Taylor's theorem to expand the function $\eta$ we get
\begin{align*}
\eta((\vts + \theta \ds)^T \zsl)
&= \eta((\vts)^T \zsl) 
    + \eta'((\vts + \bar{\theta} \ds)^T \zsl)(\theta \ds)^T \zsl
\end{align*}
, where $\bar{\theta} \in [0, \theta]$. 
Continuing from above and from Lemma \ref{lemma:eigvalue} we have,
with probability at least $1 - \Abs{S} \exp (\nicefrac{(-m \cmin)}{2 \Abs{S}})$:
\begin{align*}
&\eigmin\left(\Hess \loss(\vts + \theta \ds)\right) \\
&\quad\geq \min_{\theta \in [0, 1]} \eigmin \Bigg(\frac{1}{m} \sum_{l=1}^m \eta((\vts)^T \zsl) \zsl (\zsl)^T \\
    &\qquad + \frac{1}{m} \sum_{l=1}^m \eta'((\vts + \bar{\theta} \ds)^T \zsl)((\theta \ds)^T \zsl) \zsl (\zsl)^T \Bigg) \\
&\quad\geq \eigmin(\mat{H}^m_{SS}) - \max_{\theta \in [0, 1]} \InNormII{\mat{A}(\theta)} 
	\geq \frac{\cmin}{2} - \max_{\theta \in [0, 1]} \InNormII{\mat{A}(\theta)},
\end{align*}
where we have defined
\begin{align*}
\mat{A}(\theta) \defeq \frac{1}{m} \sum_{l=1}^m \eta'((\vts + \theta \ds)^T \zsl)(\theta \ds)^T \zsl \zsl (\zsl)^T.
\end{align*}
Next, the spectral norm of $\mat{A}(\theta)$ can be bounded as follows:
\begin{align*}
\InNormII{\mat{A}(\theta)} 
&\; \leq \max_{\NormII{\yv} = 1} \Bigg\{
    \frac{1}{m} \sum_{l=1}^m \Abs{\eta'((\vts + \theta \ds)^T \zsl)} \Abs{((\theta \ds)^T \zsl)} 
        \times \yv^T (\zsl (\zsl)^T) \yv  \Bigg\} \\
&\; < \max_{\NormII{\yv} = 1} \left\{ \frac{1}{10 m} \sum_{l=1}^m \NormI{(\theta \ds)} \NormInfty{\zsl}
    \yv^T (\zsl (\zsl)^T) \yv \right\} \\
&\; \leq \theta \max_{\NormII{\yv} = 1} \left\{ \frac{1}{10 m} \sum_{l=1}^m \sqrt{|S|} \NormII{\ds}
    \yv^T (\zsl (\zsl)^T) \yv \right\} \\
&\; = \theta b \sqrt{|S|} \InNormII{\frac{1}{10 m} \sum_{l=1}^m \zsl (\zsl)^T} \\
&\; \leq \frac{(b \sqrt{|S|} \dmax)}{5} \leq \frac{\cmin}{4},
\end{align*}
where in the second line we used the fact that $\eta'(.) < \nicefrac{1}{10}$
and in the last line we assumed that $\nicefrac{(b \sqrt{|S|} \dmax)}{5} \leq \nicefrac{\cmin}{4}$ ---
an assumption that we verify momentarily. Having upper bounded the 
spectral norm of $\mat{A}(\theta)$, we have
\begin{align}
\eigmin\left(\Hess \loss(\vts + \theta \ds)\right) \geq \frac{\cmin}{4} \label{eq:lb3}.
\end{align}
Plugging back the bounds given by \eqref{eq:lb1}, \eqref{eq:lb2} and \eqref{eq:lb3}
in \eqref{eq:f} and equating to zero we get
\begin{gather*}
- b \sqrt{|S|} \left(\nu \kappa + \sqrt{\frac{2}{m} \log \frac{2n}{\delta}}\right)
+ \frac{b^2 \cmin}{4} - \lambda \sqrt{\Abs{S}} b = 0 \\
\implies b = \frac{4 \sqrt{\Abs{S}}}{\cmin}
	\left(\lambda + \nu \kappa + \sqrt{\frac{2}{m} \log \frac{2n}{\delta}}\right).
\end{gather*} 
Finally, coming back to our prior assumption we have
\begin{align*}
b = \frac{4 \sqrt{\Abs{S}}}{\cmin}
	\left(\lambda + \nu \kappa + \sqrt{\frac{2}{m} \log \frac{2n}{\delta}}\right)
	\leq \frac{5 \cmin}{4 \sqrt{\Abs{S}} \dmax}.	
\end{align*}
The above assumption holds if the regularization parameter $\lambda$ is bounded
as follows:
\begin{align*}
\lambda \leq \frac{5\cmin^2}{16\Abs{S} \dmax} - \sqrt{\frac{2}{m} \log \frac{2n}{\delta}} - \nu \kappa.
\end{align*}
\end{proof}
\begin{lemma}
\label{lemma_l1norm_bound}
If the regularization parameter $\lambda$ satisfies the following condition:
\begin{align*}
\lambda \geq \nu \kappa + \sqrt{\frac{2}{m} \log \frac{2n}{\delta}},
\end{align*}
then we have that 
\begin{align*}
\NormI{\vh - \vt} \leq \frac{5\cmin}{\dmax}
\end{align*}
with probability at least $1 - (\delta + \Abs{S} \exp (\nicefrac{(-m \cmin)}{2 \Abs{S}}) + 
\Abs{S} \exp (\nicefrac{(-m (1 - \nu))}{4 \Abs{S}}))$.
\end{lemma}
\begin{proof}
Define $\Delta \defeq \vh - \vt$. Also for any vector $\vect{y}$ let
the notation $\vect{y}_{\overline{S}}$ denote the vector $\vect{y}$ with
the entries not in the support, $S$, set to zero, i.e.
\begin{align*}
\left[\vect{y}_{\overline{S}} \right]_i &= \left\{\begin{array}{lr}
y_i & \text{if $i \in S$},\\
0 & \text{otherwise}.
\end{array}\right.
\end{align*}
Similarly, let the notation $\vect{y}_{\overline{\Sc}}$ denote the vector $\vect{y}$
with the entries not in $\Sc$ set to zero, where $\Sc$ is the complement of $S$.
Having introduced our notation and since, $S$ is the support of the true vector $\vt$,
we have by definition that $\vt = \vtbs$. We then have, using the reverse triangle inequality,
\begin{align}
\NormI{\vh} &= \NormI{\vt + \Delta} = \NormI{\vtbs + \dbs + \dbsc} \notag \\
&= \NormI{\vtbs - (- \dbs)} + \NormI{\dbsc} \notag \\
&\geq \NormI{\vt} - \NormI{\dbs} + \NormI{\dbsc}. \label{eq:onenorm_bound1}
\end{align}
Also, from the optimality of $\vh$ for the $\ell_1$-regularized problem we have that
\begin{align}
\loss(\vt) + \lambda \NormI{\vt} &\geq \loss(\vh) + \lambda \NormI{\vh} \notag \\
\implies \lambda(\NormI{\vt} - \NormI{\vh}) &\geq \loss(\vh) - \loss(\vt). \label{eq:onenorm_bound2}
\end{align}
Next, from convexity of $\loss(.)$ and using the Cauchy-Schwartz inequality
we have that
\begin{align}
\loss(\vh) - \loss(\vt) &\geq \Grad \loss(\vt)^T(\vh - \vt) 
	\geq - \NormInfty{\Grad \loss(\vt)} \NormI{\Delta} \notag \\
&\geq - \frac{\lambda}{2} \NormI{\Delta}, \label{eq:onenorm_bound3}
\end{align}
where in the last line we used the fact that $\lambda \geq \NormInfty{\Grad \loss(\vt)}$.
Thus, we have from \eqref{eq:onenorm_bound1},
\eqref{eq:onenorm_bound2} and \eqref{eq:onenorm_bound3} that 
\begin{align}
& \frac{1}{2} \Norm{\Delta}_1 \geq \NormI{\vh} - \NormI{\vt} \notag \\
\implies &\frac{1}{2} \Norm{\Delta}_1 \geq \NormI{\dbsc} - \NormI{\dbs} \notag \\
\implies &\frac{1}{2} \NormI{\dbsc} + \frac{1}{2} \NormI{\dbs} \geq \NormI{\dbsc} - \NormI{\dbs} \notag \\
\implies &3 \NormI{\dbs} \geq \NormI{\dbsc}. \label{eq:onenorm_bound4}
\end{align}
Finally, from \eqref{eq:onenorm_bound4} and Lemma \ref{lemma_l2norm_bound} we
have that
\begin{align*}
\NormI{\Delta} &= \NormI{\dbs} + \NormI{\dbsc} \leq 4 \NormI{\dbs} \leq 4 \sqrt{\Abs{S}} \NormII{\dbs} 
	\leq \frac{5\cmin}{\dmax}.
\end{align*}
\end{proof}

Now we are ready to present our main result on recovering the true PSNE set.
\begin{theorem}
\label{thm:psne}
If for all $i$, $\Abs{S_i} \leq k$, the minimum payoff $\rhomin \geq \nicefrac{5 \cmin}{\dmax}$,
and the regularization parameter and the number of samples satisfy the following conditions:
\begin{gather}
\nu \kappa + \sqrt{\frac{2}{m} \log \frac{6n^2}{\delta}} \leq \lambda 
	\leq 2K + \nu \kappa - \sqrt{\frac{2}{m} \log \frac{6n^2}{\delta}} \label{eq:reg_param_bounds} \\
m \geq \max \Bigg\{ \frac{2}{K^2} \log \left(\frac{6n^2}{\delta}\right),	 
	\frac{2k}{\cmin} \log \left(\frac{3kn}{\delta} \right), \notag \\
		\frac{4k}{1 - \nu} \log \left(\frac{3kn}{\delta} \right) \Bigg\}, \label{eq:sample_complexity}
\end{gather}
where $K \defeq \nicefrac{5\cmin^2}{32k\dmax} - \nu \kappa$, 
then with probability at least $1 - \delta$, for $\delta \in [0, 1]$,
we recover the true PSNE set, i.e. 
$\NE(\widehat{\mat{W}}, \widehat{\vect{b}}) = \NE(\mat{W}^*, \vect{b}^*)$.
\end{theorem}
\begin{proof}
From Cauchy-Schwartz inequality and Lemma \ref{lemma_l1norm_bound} we have
\begin{align*}
\Abs{(\vhi - \vti)^T \zi} \leq \NormI{\vhi - \vti} \NormInfty{\zi} \leq \frac{5 \cmin}{\dmax}.
\end{align*}
Therefore, we have that
\begin{align*}
(\vti)^T \zi - \frac{5 \cmin}{\dmax}
        \leq \vhi^T \zi \leq (\vti)^T \zi + \frac{5 \cmin}{\dmax}.
\end{align*}
Now, if $\forall\; \vect{x} \in \NE^*$, $(\vti)^T \zi \geq \nicefrac{5 \cmin}{\dmax}$,
then $\vhi^T \zi \geq 0$.  Using an union bound argument over all players $i$, we can
show that the above holds with probability at least  
\begin{align}
1 - n(\delta + k \exp (\nicefrac{(-m \cmin)}{2k}) + 
	k\exp (\nicefrac{(-m (1 - \nu))}{4k})) \label{eq:thm_main_prob}
\end{align}
for all players. 
Therefore, we have that $\NE(\widehat{\mat{W}}, \widehat{\vect{b}}) = \NE^*$ 
with high probability. Finally, setting $\delta = \nicefrac{\delta'}{3n}$,
for some $\delta' \in [0, 1]$, and ensuring that the last two terms in \eqref{eq:thm_main_prob}
are at most $\nicefrac{\delta'}{3n}$ each, we prove our claim.
\end{proof}
To better understand the implications of the theorem above, 
we discuss some possible operating regimes for learning sparse linear influence games
in the following paragraphs.
\begin{remark}[Sample complexity for fixed $q$]
In the theorem above, if $q$ is constant, which in turn makes $\nu$ constant,
then $K = \BigOm{\nicefrac{1}{k^2}}$, and the sample complexity of learning sparse
linear games grows as $\BigO{k^4 \log n}$.
However, if $q$ is small enough such that $\nu \leq \nicefrac{1}{k}$, then the constant
$\dmax$ is no longer a function of $k$ and hence
$K = \BigOm{\nicefrac{1}{k}}$. Therefore, the sample complexity scales as
$\BigO{k^2 \log n}$ for exact PSNE recovery.
\end{remark}
The sample complexity of $\BigO{\poly(k) \log n}$ for exact recovery of the PSNE set
can be compared with the sample complexity of $\BigO{kn^3 \log^2 n}$ for the maximum likelihood estimate (MLE)
of the PSNE set as obtained by Honorio \cite{Honorio2016}. Note that while the MLE procedure is consistent,
i.e. MLE of the PSNE set is equal to the true PSNE set with probability converging to 1 as
the number of samples tend to infinity, it is NP-hard \footnote{Irfan and Ortiz \cite{irfan14}
showed that counting the number of Nash equilibria is $\#$P-complete. Therefore, computing the
log-likelihood is NP-hard.}.
In contrast, the logistic regression method
is computationally efficient. Further, while the sample complexity of our method
seems to be better than the empirical log-likelihood minimizer as given by Theorem 3 in \cite{Honorio2016},
in this paper we restrict ourselves to LIGs with strictly positive payoff in the PSNE set ---
such games are invariably easier to learn than general LIGs considered by \cite{Honorio2016}.

Honorio \cite{Honorio2016} also obtained lower bounds on the number of samples required by any conceivable 
method, for exact PSNE recovery, by scaling the parameter $q$ with the number of players.
Therefore, in order to compare our results with the information-theoretic limits of learning LIGs,
we consider the regime where the parameter $q$ scales with the number of players $n$.
\begin{remark}[Sample complexity for $q$ varying with number of players]
If we consider the regime where the signal level scales as $q = \nicefrac{(\Abs{\NE^*} + 1)}{2^n}$.
Then, $\dmax = \BigO{\nicefrac{k}{(2^n - \Abs{\NE^*})}}$, 
and as a result $K = \BigOm{\nicefrac{2^n}{k^2}}$.
Therefore, the sample complexity, which is dominated by the second and third terms in 
\eqref{eq:sample_complexity}, is given as $\BigO{k \log (kn)}$.
In general if $q = \BigT{\exp(-n)}$, then the sample complexity
for recovering the PSNE set exactly is $\BigO{k \log (kn)}$.
\end{remark}
Once again we observe that even in the regime of $q$ scaling exponentially with $n$,
the sample complexity of $\BigO{k \log (kn)}$
is better than the information theoretic limit of $\BigO{kn \log^2 n}$, by a factor of $\BigO{n \log n}$. 
This can be attributed to the fact that we consider restricted ensembles of
games with strictly positive payoffs in the PSNE set, as opposed to general LIGs.

Further, from the aforementioned remarks we see that as the signal level $q$ decreases,
the sufficient number of samples needed to recover the PSNE set reduces; 
and in the limiting case of $q$ decreasing exponentially with the number
of players, the sample complexity scales as $\BigO{k \log (kn)}$. This seems
counter-intuitive --- with increased signal level, a learning problem should
become easier and not harder. To understand this seemingly counter-intuitive behavior,
first observe that the constant $\nicefrac{\dmax}{\cmin}$ can be thought of as the
``condition number'' of the loss function given by \ref{eq:loss}, Therefore, the sample complexity
as given by Theorem \ref{thm:psne} can be written as $\BigO{k^2 (\nicefrac{\dmax}{\cmin})^2 \log n}$.
From \eqref{eq:hess_comb}, we see that as the signal level increases, the Hessian of the loss
becomes more ill-conditioned, since the data set now comprises of many repetitions of the few
joint-actions that are in the PSNE set; thereby increasing the dependency ($\dmax$) between actions of players
in the sample data set.

\section{Experiments}
\label{sec:experiments}
\begin{figure}[t]
\centering
\includegraphics[width=\linewidth]{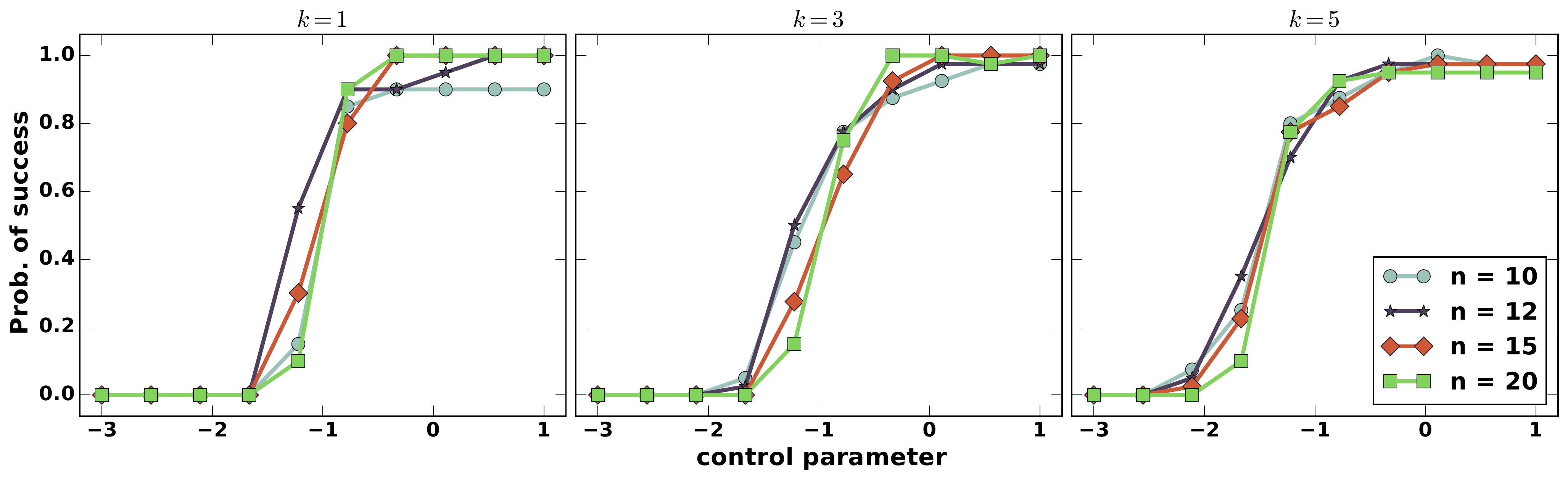}
\caption{The probability of exact recovery of the PSNE set computed across $40$ randomly sampled
LIGs as the number of samples is scaled as $\lfloor (C) (10^{c}) (k^2 \log(\nicefrac{6n^2}{\delta}))\rfloor$,
where $c$ is the control parameter and the constant $C$ is $10000$ for the $k = 1$ case 
and $1000$ for the remaining two case.%
\label{fig:simulation}}
\end{figure}%
In order to verify that our results and assumptions indeed hold in practice,
we performed various simulation experiments. We generated random LIGs for $n$ players
and exactly $k$ neighbors by first creating a matrix $\mat{W}$ of all zeros and then
setting $k$ off-diagonal entries of each row, chosen uniformly at random, to $-1$.
We set the bias for all players to $0$.
We found that any odd value of $k$ produces games with strictly
positive payoff in the PSNE set. Therefore, for each value of $k$ in $\{1, 3, 5\}$,
and $n$ in $\{10, 12, 15, 20\}$, we generated $40$ random LIGs.
The parameters $q$ and $\delta$ were set to the constant value of $0.01$
and the regularization parameter $\lambda$ was set according to Theorem \ref{thm:psne}
as some constant multiple of $\sqrt{(\nicefrac{2}{m}) \log (\nicefrac{2n}{\delta})}$.
Figure \ref{fig:simulation} shows the probability of successful recovery of the PSNE,
for various combinations of $(n, k)$, where the probability was computed
as the fraction of the $40$ randomly sampled LIGs for which the learned PSNE set matched the
true PSNE set exactly.
For each experiment, the number of samples was computed as:
$\lfloor(C) (10^{c}) (k^2 \log(\nicefrac{6n^2}{\delta}))\rfloor$, where $c$ is the control parameter 
and the constant $C$ is $10000$ for $k = 1$ and $1000$ for $k = 3$ and $5$. Thus,
from Figure \ref{fig:simulation} we see that, the sample complexity of $\BigO{k^2 \log n}$ as given
by Theorem \ref{thm:psne} indeed holds in practice --- i.e. there exists constants
$c$ and $c'$ such that if the number of samples is less than $c k^2 \log n$, we
fail to recover the PSNE set exactly with high probability, while if the number of 
samples is greater than $c' k^2 \log n$ then we are able to recover the PSNE set
exactly, with high probability. Further, the scaling remains consistent as the
number of players $n$ is changed from $10$ to $20$.

\section{Conclusion}
\label{sec:conclusion}
In this paper, we presented a computationally efficient and statistically consistent 
method, based on $\ell_1$-regularized logistic regression, for learning linear influence games --- 
a subclass of parametric graphical games with linear payoffs. Under some mild conditions
on the true game, we showed that as long as the number of samples scales
as $\BigO{\poly(k) \log n}$, where $n$ is the number of players and $k$ is
the maximum number of neighbors of any player; then we can recover
the pure-strategy Nash equilibria set of the true game in polynomial time and with probability
converging to 1 as the number of samples tend to infinity. An interesting direction for future work
would be to consider structured actions --- for instance
permutations, directed spanning trees, directed acyclic graphs among others --- thereby extending
the formalism of linear influence games to the structured prediction setting. A more technical extension
would be to consider a local noise model where the observations are drawn from the PSNE set but with each action
independently corrupted by some noise. Other ideas that might be worth pursuing are: considering mixed strategies,
correlated equilibria and epsilon Nash equilibria, and incorporating latent or unobserved
actions and variables in the model.

\bibliographystyle{IEEEtran}
\bibliography{paper}

\end{document}